\documentclass[11pt,a4paper]{article}

\usepackage[utf8]{inputenc} 
\usepackage[T1]{fontenc}    %
\usepackage[english]{babel} 
\usepackage[final]{microtype} 

\usepackage{amssymb,amsmath,mathtools} 
\usepackage{amsthm} 

\usepackage{xcolor}
\usepackage{bbm}
\usepackage{braket}

\usepackage[hmargin=0.12\paperwidth,vmargin=0.16\paperwidth,bindingoffset=0cm]%
{geometry} 

\numberwithin{equation}{section} 

\theoremstyle{plain}
  
  \newtheorem{prop}{Proposition}
  
  \newtheorem{cor}{Corollary}
\theoremstyle{definition}
  
  \newtheorem{remark}{Remark}
  
  \numberwithin{prop}{section}
   \numberwithin{cor}{section}
   \numberwithin{remark}{section}

\newcommand{\sgn}{\mathrm{sgn}}

\title{\Large\bfseries Integrability enabled computations relating to the fixed trace Laguerre ensemble}

\author{Peter J. Forrester${}^1$ and Shinsuke M. Nishigaki${}^2$}
\date{}

\begin{document}

\maketitle

${}^1$ School of Mathematics and Statistics,  The University of Melbourne,
Victoria 3010, Australia. \: \: Email: {\tt pjforr@unimelb.edu.au}; \\

${}^2$ Graduate School of Natural Science and Technology, Shimane University, \\
Matsue 690-8504, Japan.  \: \: Email: {\tt nishigaki@riko.shimane-u.ac.jp}

\date{}


\begin{abstract}
Studies of density matrices for random quantum states lead naturally to the fixed
trace Laguerre ensemble in random matrix theory. Previous studies have uncovered
explicit rational function formulas for moments of purity statistic (trace of
the squared density matrix), and also a third order linear differential equation satisfied
by the eigenvalue density. We further probe the origin of these results from
the viewpoint of integrability, which  is taken here to mean wider classes of recursions
and differential equations, and give extensions. Prominent in our study are first order linear
matrix differential equations. One application given is to the derivation of the third
order scalar equation for the density.  Another is to obtain the explicit rational
function formula for the variance of the purity statistic in the $\beta$ generalised fixed
trace Laguerre ensemble. In the original case ($\beta = 2$), the purity cumulants are 
expressed in terms of the large argument expansion of a particular
$\sigma$-Painlev\'e IV transcendent. In a different but related direction, the exact computation of the two-point correlation for the
fixed determinant circular unitary ensemble SU$(N)$ is given the Appendix.
\end{abstract}

 {\it Dedicated to the memory of Santosh Kumar and his work on exact results in RMT and their applications\footnote{For the context of the present work in the interests and contributions of Santosh Kumar in random matrix theory (RMT), the reader is referred to the tribute and review article \cite{Fo25}.}}

\maketitle

\section{Introduction}

By definition, a standard Gaussian complex random matrix $G$ has all
entries independent and identically distributed as standard complex random
variables ${\rm N}[0,1/\sqrt{2}] + i {\rm N}[0,1/\sqrt{2}]$ --- here the adjective
``standard" refers to the mean of the squared modulus equalling unity. The
Gaussian unitary ensemble (GUE) of complex Hermitian matrices is constructed out
of the matrices $\{G \}$ by forming ${1 \over 2} ( G + G^\dagger)$. The corresponding
distribution on the space of Hermitian matrices is proportional to $e^{- {\rm Tr} \, G^2}$.
Changing variables to the eigenvalues $\{ \lambda_j^2 \}_{j=1}^N$ 
(here we are assuming $G$ is of size $N \times N$) 
and the eigenvectors gives for the eigenvalue probability density function (PDF)
\begin{equation}\label{0.1}
{1 \over C_N^{(\rm G)}} \prod_{l=1}^N e^{- \lambda_l^2} \prod_{1 \le j < k \le N} (\lambda_k - \lambda_j)^2,
\end{equation}
supported on $\lambda_l \in \mathbb R$ ($l=1,\dots,N$), where the normalisation $C_N^{(\rm G)}$ is
given by
\begin{equation}\label{0.1a}
C_N^{(\rm G)} = \pi^{N/2} 2^{- N (N - 1)/2}  \prod_{l=1}^N l !;
\end{equation}
see e.g.~\cite[Prop.~1.3.4 with $\beta = 2$]{Fo10}.

Suppose instead that the size of $G$ is $n \times N$, $n \ge N$, also referred to as a rectangular complex Ginibre matrix
\cite{BF25}. The Laguerre unitary ensemble (LUE) of
complex Hermitian positive definite matrices is constructed out of $\{ G\}$ by forming $G^\dagger G$, and
the corresponding eigenvalue PDF is given by
\begin{equation}\label{0.2}
p_{N,a}^{(\rm L)}(\lambda_1,\dots,\lambda_N) := {1 \over C_{N,a}^{(\rm L)}} \prod_{l=1}^N \lambda_l^a e^{- \lambda_l} \prod_{1 \le j < k \le N} (\lambda_k - \lambda_j)^2, \quad a := n - N,
\end{equation}
supported on $\lambda_l \in \mathbb R^+$ ($l=1,\dots,N$), where the normalisation $C_{N,a}^{(\rm L)}$ is
given by
\begin{equation}\label{0.2a}
C_{N,a}^{(\rm L)} = N!  \prod_{l=0}^{N-1} \Gamma(l+1) \Gamma(a+l+1);
\end{equation}
see e.g.~\cite[Prop.~3.2.2 with $\beta = 2$, $m \mapsto N$]{Fo10}.
The fixed trace (chosen to equal unity) Laguerre unitary ensemble (fLUE) is obtained by the
same construction involving $\{ G\}$, now multiplied by a scalar factor $1/{\rm Tr} \, G^\dagger G$. The
eigenvalue PDF is
\begin{equation}\label{0.3}
p_{N,a}^{(\rm fL)}(\lambda_1,\dots,\lambda_N) :=
{1 \over C_{N,a}^{(\rm fL)}}  \delta \Big (1 -  \sum_{l=1}^N \lambda_l \Big ) \prod_{l=1}^N \lambda_l^a \prod_{1 \le j < k \le N} (\lambda_k - \lambda_j)^2, \quad a := n - N,
\end{equation}
supported on $\lambda_l \in (0,1)$ ($l=1,\dots,N$), 
where the normalisation $C_{N,a}^{(\rm fL)}$ is
given by
\begin{equation}\label{0.3a}
C_{N,a}^{(\rm fL)} = {N! \over \Gamma(aN +N^2) }\prod_{l=0}^{N-1} \Gamma(l+1) \Gamma(a+l+1);
\end{equation}
see e.g.~\cite[Eq.~(4.155) with $\beta = 2$,  $m \mapsto N$]{Fo10}. Note that the eigenvalues 
must sum to unity as indicated by the Dirac delta function.

The eigenvalue PDF (\ref{0.3}) is prominent in the study of the bipartite entanglement of
a finite basis quantum mechanical system ; see e.g.~\cite{BZ06}.
The two subsystems, denoted $A$ and $B$ say, are taken to have dimensions
$n$ and $N$ and orthonormal bases $\{|a_i\rangle \}_{i=1,\dots,n}$, $\{|b_i\rangle \}_{i=1,\dots,N}$ 
respectively. Hence a state $|\psi \rangle$ can be written
\begin{equation}\label{2.vag2}
| \psi \rangle = \sum_{i=1}^n \sum_{j=1}^N x_{i,j} |a_i \rangle \otimes
 |b_j \rangle.
\end{equation}
On the other hand, according to the Schmidt decomposition, for a given $| \psi \rangle$,
there exists orthonormal vectors $\{ |v_i^A\rangle \}_{i=1,\dots,N}$ and
$\{ |v_i^B\rangle \}_{i=1,\dots,N}$ relating to the reduced density matrices of the two subsystems
such that 
\begin{equation}\label{2.vag3}
|\psi \rangle = \sum_{i=1}^N \sqrt{\lambda_i} |v_i^A \rangle \otimes |v_i^B \rangle, \quad \sum_{j=1}^N \lambda_j = 1;
\end{equation}
see e.g.~\cite[\S 3.3.4]{Fo10}. Under the assumption that the 
coefficients $\{x_{i,j}\}$ in (\ref{2.vag2}) are independent standard complex Gaussians
so that $[x_{i,j}] = G$ 
as specified above (\ref{0.2}),
up to the constraint $\sum_{i=1}^n \sum_{j=1}^N | x_{ij} |^2 = {\rm Tr} \, G^\dagger G = 1$
to ensure the normalisation $\langle \psi | \psi \rangle = 1$, the distribution of the coefficients $\{\lambda_i \}$
in (\ref{2.vag3}) is given by (\ref{0.3}).

Quantities used to quantify the bipartite entanglement include the purity, defined as $\sum_{j=1}^N \lambda_j^2$
(i.e.~the trace of the squared reduced density matrix for subsystems $A$ or $B$;
see e.g.~\cite{BZ06}) and the von Neumann entropy $-\sum_{j=1}^N \lambda_j \log \lambda_j$
(i.e.~minus the trace of the  reduced density matrix times the logarithm of the reduced density matrix).
Generalising both is the so-called quantum Tsallis entropy $(1/(1-q))(1- \sum_{j=1}^N \lambda_j^q)$
(set $q=2$  to relate to the purity, and take the limit $q \to 1$ to obtain the
 von Neumann entropy).
All these quantities are examples of linear statistics, being of the form 
\begin{equation}\label{Aa}
A = \sum_{j=1}^N a(\lambda_j)
\end{equation}
 for some
function of a single variable $a(x)$.

Consider a random state  as specified below (\ref{2.vag3}). Let $\rho_{(1)}^{(\rm fL)}(x)$
denote the density of eigenvalues corresponding to the PDF (\ref{0.3}).
With respect to the random state, for the linear statistic $A$ we have
\begin{equation}\label{Aa1}
\langle  A \rangle = \int_I a(x)  \rho_{(1)}^{(\rm fL)}(x) \, dx,
\end{equation}
where $I$ denotes the interval of support of the density.  An expression for $\rho_{(1)}^{(\rm fL)}(x)$
is given in \cite[Eq.~(31)]{KP11a}, which involves the difference of two single sums from 0 up to $N-1$ of Gauss
hypergeometric functions. Another form, involving a sum over $\{(x/(1 - x))^l \}_{l=0}^{2N-2}$, where the
coefficients themselves are further sums from 0 up to $N$, is known from
\cite{ATK09,Vi10,VPO16} and can also be found in \cite{KP11a}. With a theme based on
computations relating to (\ref{0.3}) enabled by integrability --- specifically recursions and differential
equations ---
of particular  interest to us
 is the fact, deduced in \cite{ATK11}, that $\rho_{(1)}^{(\rm fL)}(x)$ satisfies the third order
linear differential equation
\begin{equation}\label{Aa1X}
\Big ( a_3(x) {d^3 \over d x^3} + a_2(x)  {d^2 \over d x^2} + a_1(x)  {d \over d x} + a_0(x) \Big ) f(x) = 0,
\end{equation}
where
\begin{align}\label{Aa1Y}
a_0(x) & = x (2N + a) (N^2 + a N - 2) - a^2 \nonumber \\
a_1(x) & = -x \Big ( x^2 (N^2 + a N - 4)(N^2 + a N - 3) - x (2N + a) ( 2 N^2 + 2 a N - 7) + a^2 - 2 \Big )
\nonumber  \\
   a_2(x) & =
   2 x^2 \Big (x^2 (N^2 + N a - 4) -x (2N + a) + 2 \Big )
    \nonumber \\
   a_3(x) & =
   x^3 \left(1- {x^2}\right).
   \end{align}
 As an application, it was shown in  \cite{ATK11} that this differential equation is
 well suited to a large $N$ asymptotic analysis using the WKB method, resulting in 
 asymptotic formulas which accurately reproduce the finite $N$ graph of $\rho_{(1)}^{(\rm fL)}(x)$.
 We will show in Section \ref{S2.1} that the differential equation (\ref{Aa1X}) can be derived from
 the third order linear differential equation already known from the density of the Laguerre
 unitary ensemble as determined by the eigenvalue PDF (\ref{0.2}) \cite{GT05,ATK11,RF19}. 
  A viewpoint of (\ref{Aa1X}) based on a first order $3 \times 3$ matrix differential equation is given in 
 Section \ref{S2.2}. Some comparative features of $\rho_{(1)}^{(\rm L)}(x)$ and $\rho_{(1)}^{(\rm fL)}(x)$
 are made in Section \ref{S2.3}.

The distribution of the purity for the random state specified below (\ref{2.vag3}) is given by
\begin{equation}\label{d1}
P_{N,a}(t) = {1 \over C_{N,a}^{(\rm fL)}}  \int_0^1 d\lambda_1 \cdots \int_0^1 d\lambda_N \,  
\delta \Big ( t -  \sum_{j=1}^N \lambda_j^2 \Big )
 \delta \Big (1 -  \sum_{l=1}^N \lambda_l \Big ) \prod_{l=1}^N \lambda_l^a \prod_{1 \le j < k \le N} (\lambda_k - \lambda_j)^2.
\end{equation}
For general $N$ the exact calculation of (\ref{d1}) is from a practical viewpoint intractable, with its complexity increasing with increasing $N$ and $a$. However, some exact results for small $N$ and or $a$ are available.
Thus its
evaluation is
almost immediate for $N = 2$, being given by \cite[Eq.~(9)]{Gi07a}
\begin{equation}\label{d1a}
P_{2,a}(t) = {(2a+3)! \over 2^{a+1} (a+1)! a!} (1 - t)^{a} \sqrt{2t - 1},
\end{equation}
supported on $1/2 < t < 1$. For $a=0$ and $N = 3$, the distribution is piecewise real analytic, 
involving only elementary functions, on the intervals
$1/3 < t < 1/2$ and $1/2 < t < 1$; see \cite[Eq.~(16)]{Gi07a}. A still more complicated
expression, although again involving only elementary functions, holds for
$N=3$ and general $a \in \mathbb Z^+$; see  \cite[Eq.~(31)]{Gi07a}, and
for $N=4$ and $a=0$; see  \cite[Eq.~(24)]{Gi07a}.

Progress for general $N$ is possible by considering
the Laplace-Fourier transform of (\ref{d1}),
\begin{equation}\label{d1+}
\hat{P}_{N,a}^{\rm (fL)}(s) = {1 \over C_{N,a}^{(\rm fL)}}  \int_0^1 d\lambda_1 \cdots \int_0^1 d\lambda_N \,  
 \delta \Big (1 -  \sum_{l=1}^N \lambda_l \Big ) \prod_{l=1}^N e^{-s   \lambda_j^2}
 \lambda_l^a \prod_{1 \le j < k \le N} (\lambda_k - \lambda_j)^2,
\end{equation}
or equivalently the moments
\begin{equation}\label{d1a+}
\Big \langle \Big (   \sum_{j=1}^N \lambda_j^2 \Big )^k \Big \rangle^{\rm (fL)},
\end{equation}
$k=1,2,\dots$; thus (\ref{d1+}) is the exponential generating function for
(\ref{d1a+}).  In particular, we know from \cite[Eq.~(10), corrected by a factor of $N!$]{Gi07} and
\cite[Eq.~(17)]{Gi07a} that
\begin{multline}\label{d1b}
\Big \langle \Big (   \sum_{j=1}^N \lambda_j^2 \Big )^k \Big \rangle^{\rm (fL)} =
{N! (N (N + a) - 1)! \over ( N(N + a) + 2k - 1)!}  \\ \times \sum_{k_1,\dots,k_N \ge 0 \atop
k_1 + \cdots + k_N = k}
{k! \over k_1! \cdots k_N!} \prod_{i=1}^N
{(N + a + 2k_i - i)! \over (N + a - i)! i!}
\prod_{1 \le i < j \le N} (2 k_i - i - 2k_j + j).
\end{multline}

To be addressed in Section \ref{S3.1} of the present paper is a relationship between
the moments (\ref{d1a+}) and  a
particular $\sigma$-Painlev\'e IV transcendent. This in turn comes about through a
relationship between the LUE version of (\ref{d1+})
\begin{equation}\label{d1L}
\hat{P}_{N,a}^{\rm (L)}(s) = {1 \over C_{N,a}^{(\rm L)}}  \int_0^\infty d\lambda_1 \cdots \int_0^\infty d\lambda_N \,  
\prod_{l=1}^N e^{-\lambda_l -s   \lambda_l^2}
 \lambda_l^a \prod_{1 \le j < k \le N} (\lambda_k - \lambda_j)^2,
\end{equation}
and the GUE average
\begin{equation}\label{d1aG}
\Big \langle    \prod_{j=1}^N( \lambda_j - 1/(2 \sqrt{s}) )^a \chi_{\lambda_l > 1/(2 \sqrt{s})}  \Big \rangle^{\rm (G)},
\end{equation}
where $\chi_A = 1$ if $A$ is true, $\chi_A = 0$ otherwise. Thus we have \cite[minor generalisation of (82)]{VPO16}
\begin{equation}\label{G1}
\hat{P}_{N,a}^{\rm (L)}(s) = {C_N^{\rm (G)} \over C_{N,a}^{(\rm L)}} e^{N/4s} (1 / \sqrt{s} )^{N  a  + N^2}
\Big \langle    \prod_{j=1}^N( \lambda_j - 1/(2 \sqrt{s}) )^a \chi_{\lambda_j > 1/(2 \sqrt{s})}  \Big \rangle^{\rm (G)},
\end{equation}
while the average (\ref{d1aG}) is known to relate to a $\tau$-function  in the theory of
Painlev\'e IV \cite{FW00} (see also \cite[Ch.~8]{Fo10}), and so satisfies a particular second order nonlinear differential equation.
The large argument Laurent solution of the latter systematically generates the LUE version of the moments
(\ref{d1a+}), and thus the moments (\ref{d1b}) themselves, as they are related by a simple proportionality
(see (\ref{ds1e+})). In Section \ref{S3.2} we make use of a known integration formula involving the Schur
polynomial weighted against the PDF (\ref{0.2}) for purposes of  generalising the formula (\ref{d1b}) so that it holds for
\begin{equation}\label{1.16a}
\Big \langle \Big (   \sum_{j=1}^N \lambda_j^q \Big )^k \Big \rangle^{\rm (fL)},
\end{equation}
where $q$ is a non-negative integer.

First order linear matrix differential equations, introduced in Section \ref{S2.2} in relation to a viewpoint on
the origin of the differential equation (\ref{Aa1X}), are derived again in Section \ref{S4}. This is for the
purpose of obtaining recurrences for the power series expansion of (\ref{d1L}) with the exponent
in the product over differences generalised from the value 2 to a general parameter $\beta$. As a result we
are able to explicitly determine the $k=1$ and $k=2$ moments (\ref{d1a+}) in the general $\beta$ case.
Motivated by the analogue of a fixed trace Hermitian matrix for unitary matrices
 being a unit determinant contraint, we provided in Appendix A the exact functional form 
of the two-point correlation for  Haar distributed SU$(N)$. .

\section{Differential equations for the density of the ${\rm fLUE}$}
\subsection{Transforming the scalar differential equation for the LUE density}\label{S2.1}
Generally, for an eigenvalue PDF $p_N(\lambda_1,\dots,\lambda_N)$, supported on $\lambda_l \in I$, $(l=1,\dots,N)$,
the eigenvalue density $\rho_{(1)}(x)$ is given by
\begin{equation}\label{u1}
\rho_{(1)}(x) = N \int_I d\lambda_2 \cdots \int_I d\lambda_N \, p_N(x,\lambda_2,\dots,\lambda_N).
\end{equation}
For the
Laguerre
 unitary ensemble we have $p_N = p_N^{( \rm L )}$ as specified by (\ref{0.2}). Previous work  \cite{GT05,ATK11,RF19}
 has shown that the corresponding density, $\rho_{(1)}^{(L)}(x)$ say, satisfies the third order linear differential
 equation
 \begin{equation}\label{u2}
\bigg ( x^3\frac{\mathrm{d}^3}{\mathrm{d}x^3}+4x^2\frac{\mathrm{d}^2}{\mathrm{d}x^2}-\left[x^2-2(a+2N)x+a^2-2\right]x\frac{\mathrm{d}}{\mathrm{d}x}+\left[(a+2N)x-a^2\right] \bigg ) f(x) = 0.
 \end{equation}
Our first objective is to show that the differential equation (\ref{Aa1X}) can be deduced as a corollary of (\ref{u2}).

To begin, generalise (\ref{0.3}) by introducing a parameter $t$ into the delta function,
\begin{equation}\label{0.3at}
p_{N,a}^{(\rm fL)}(\lambda_1,\dots,\lambda_N;t) :=
{1 \over C_{N,a}^{(\rm fL)}}  \delta \Big (t -  \sum_{l=1}^N \lambda_l \Big ) \prod_{l=1}^N \lambda_l^a \prod_{1 \le j < k \le N} (\lambda_k - \lambda_j)^2.
\end{equation}
The corresponding Laplace-Fourier transform is
\begin{equation}\label{0.3bt}
\hat{p}_{N,a}^{(\rm fL)}(\lambda_1,\dots,\lambda_N;s) :=
{1 \over C_{N,a}^{(\rm fL)}} \prod_{l=1}^N e^{-s \lambda_l} \lambda_l^a \prod_{1 \le j < k \le N} (\lambda_k - \lambda_j)^2,
\end{equation}
which up to a scaling of the eigenvalues, and the normalisation, is just the eigenvalue PDF for the Laguerre
unitary ensemble (\ref{0.2}).

Now set
\begin{equation}\label{0.3c}
\rho_{(1)}^{(\rm fL)}(x;t) = N \int_I d\lambda_2 \cdots \int_I d\lambda_N \, p_N^{(\rm fL)}(x,\lambda_2,\dots,\lambda_N;t).
\end{equation}
Taking the Laplace-Fourier transform, making use of (\ref{0.3bt}), shows
\begin{align*}
\hat{\rho}_{(1)}^{(\rm fL)}(x;s)  & = N \int_0^\infty d \lambda_2 \cdots \int_0^\infty d \lambda_N \, \hat{p}_{N,a}^{(\rm fL)}(x,\lambda_2,\dots,\lambda_N;s) \\
& =  { C_{N,a}^{(\rm L)} \over C_{N,a}^{(\rm fL)} }
{1 \over s^{N a + N^2 - 1}} {\rho}_{(1)}^{(\rm  L)} (sx).
\end{align*}
Taking the inverse transform and setting $t=1$ then implies
\begin{align}\label{0.3d}
{\rho}_{(1)}^{(\rm fL)}(x) & =   { C_{N,a}^{(\rm L)} \over C_{N,a}^{(\rm fL)} } {1 \over 2 \pi i} \int_{c - i \infty}^{c + i \infty}
{  {\rho}_{(1)}^{(\rm  L)} (sx) \over s^{N a + N^2 - 1}} e^s \, ds, \quad c > 0 \nonumber \\
& =   { C_{N,a}^{(\rm L)} \over C_{N,a}^{(\rm fL)} } x^{N a + N^2 - 2}   {1 \over 2 \pi i} \int_{\tilde{c} - i \infty}^{\tilde{c} + i \infty}
{  {\rho}_{(1)}^{(\rm  L)} (s) \over s^{N a + N^2 - 1}} e^{s/x} \, ds, \quad \tilde{c} > 0.
\end{align}
We remark that the use of the Laplace-Fourier transform as outlined above can be found in nearly all\footnote{An exception is \cite{ACV18}, which uses a Gaussian approximation to the delta function.}
analytic studies relating to the fixed trace Laguerre ensemble, or indeed general random matrix
ensembles with a fixed trace condition; see \cite{AV00,ATK11,KP11a,CLZ10,Vi11,AV11,AC18,Ku19,FK19} as a selection of such studies.

The relationship (\ref{0.3d}) allows the differential equations (\ref{Aa1X}) and (\ref{u2}) to be linked.

\begin{prop}
The fact that $\rho_{(1)}^{(\rm L)}(s)$ in (\ref{0.3d}) satisfies the differential equation (\ref{u2}) (with $x$ replaced
by $s$ throughout)
implies that $\rho_{(1)}^{(\rm fL)}(x)$ satisfies (\ref{Aa1X}).
\end{prop}

\begin{proof}
Begin with the differential equation (\ref{u2}), and replace $x$ 
by $s$ throughout. Multiply through by $e^{s/x} / s^{N a + N^2 - 1} $ and integrate
over $s$ in the complex plane from $\tilde{c} - i \infty$ to $\tilde{c} + i \infty$, $ \tilde{c} > 0$.
Setting
$$
g_p(x) =  \int_{\tilde{c} - i \infty}^{\tilde{c} + i \infty} {
s^p e^{s/x} \over s^{N a + N^2 - 1}}  \rho_{(1)}^{(\rm L)}(s) \, ds
$$
and considering for example this procedure applied to the first term in
(\ref{u2}), upon 
integration by parts  we obtain the identity
\begin{multline*}
 \int_{\tilde{c} - i \infty}^{\tilde{c} + i \infty} {
s^3 e^{s/x} \over s^{N a + N^2 - 1} } {d^3 \over d s^3}  \rho_{(1)}^{(\rm L)}(s) \, ds  \\
 =
- (2 - aN -N^2) (3 - aN - N^2) (4 - aN - N^2) g_0(x) - {1 \over x^3} g_3(x)    \\
  - {3 \over x^2} (4 - aN - N^2)  g_2(x) 
  - {3 \over x} (3 - aN - N^2) (4 - aN - N^2) g_1(x).
\end{multline*}
This becomes a differential identity for $g_0(x)$ upon noting
\begin{equation}\label{gg}
g_p(x) = \Big (-x^2 {d \over dx} \Big )^p g_0(x).
\end{equation}
The above procedure applied to each of the terms in (\ref{u2}) similarly gives
rise to further differential identities for $g_0(x)$. Adding them together tells us that $g_0(x)$ satisfies
a particular third order linear differential equation.

Writing $g_0(x) = h(x)/x^{Na+N^2-2}$, the latter transforms to a third order linear differential equation
for $h(x)$, which according to (\ref{0.3d}) is proportional to ${\rho}_{(1)}^{(\rm fL)}(x)$.
Writing the transformed third order linear differential equation in the form 
(\ref{Aa1X}), we see that the values of the coefficients (\ref{Aa1Y}) result (the necessary simplification
is best done using computer algebra).
\end{proof}

\subsection{Matrix differential equation viewpoint}\label{S2.2}
In addition to $\rho_{(1)}^{(\rm L)}(x)$ satisfying a third order linear
scalar differential equation, it also satisfies a first order $3 \times 3$
matrix differential equation. This is a fact known explicitly for the
density in the cases of the Gaussian unitary ensemble 
(\ref{0.1})
and Jacobi
unitary ensemble (up to normalisation the eigenvalue PDF of the latter is given
by (\ref{0.1}) with $e^{-\lambda_l^2}$ replaced by
$\lambda_l^a (1 - \lambda_l)^b$, $0 < \lambda_l < 1$). Here we will first make the matrix
differential equation explicit. We will then show that this implies that 
$\rho_{(1)}^{(\rm fL)}(x)$ also satisfies a first order $3 \times 3$
matrix differential equation, and that this in turn implies the third order linear
scalar differential equation (\ref{Aa1X}).

The matrix differential equation relating to $\rho_{(1)}^{(\rm L)}(x)$ is actually
a special case of a more general matrix differential equation applying to the
density of the Laguerre $\beta$-ensemble, the latter characterised by having an
eigenvalue PDF proportional to
\begin{equation}\label{0.2a+}
 \prod_{l=1}^N \lambda_l^{  a}  e^{- \beta \lambda_l/2} \prod_{1 \le j < k \le N} |\lambda_k - \lambda_j|^\beta, 
\end{equation}
which each $\lambda_l \in \mathbb R^+$.
It is known from
\cite{Fo94b} (see also \cite[\S 13.2.5]{Fo10}) that for $\beta$ even
\begin{equation} \label{eq:r1}
	\rho_{N,\beta}^{(\rm L)} (x) = N { W_{ 2a/\beta+2, \beta, N-1 } \over W_{2a/\beta, \beta, N} }
	x^a e^{-\beta x/2} \,_1 F_1^{(\beta/2)}(-N+1; 2a/\beta +2; (x)^\beta ),  
\end{equation}
where $\,_1 F_1^{(\beta/2)}$ --- technically a particular multidimensional hypergeometric function based on on Jack polynomials \cite[\S 13.1.1]{Fo10}--- has the multi-dimensional integral representation
\begin{multline} \label{eq:r2}
	\,_1 F_1^{(\beta/2)}(-N+1; 2a/\beta +2; (x)^\beta ) =
	  {1 \over M_\beta (2a/\beta + 2/\beta -1, N-1, 2/\beta)}   {1 \over (2 \pi )^N} 
	\int_{- \pi }^{\pi} d \theta_1 \cdots
	\\
	\times  \int_{-\pi}^{\pi} d\theta_\beta \, \prod_{l=1}^\beta e^{ i \theta_l (2a/\beta + 2/\beta -1)}
	(1+e^{- i \theta_l})^{N + 2a/\beta + 2/\beta -2} e^{-x e^{i  \theta_l} }
	\prod_{1 \leq j<k \leq \beta} | e^{i \theta_k} -e^{ i \theta_j}|^{4/\beta}.
\end{multline}
The quantities $W_{\alpha, \beta, n}$ and $M_n(\lambda_1,\lambda_2,\tau)$ are
normalisations with known gamma function evaluations --- their precise form plays no
role in specifying differential equations --- also $(x)^\beta$ is a multivariate
notation, indicating that the argument $x$ is repeated $\beta$ times.

To specify the matrix differential equation, knowledge of a differential-difference system satisfied by
(\ref{eq:r2}) is required. As a preliminary for this purpose,
let $e_p(y_1,\dots,y_N) $ denote the elementary symmetric polynomials in $\{ y_j \}_{j=1}^N$, and
define
\begin{multline}\label{5.1}
J_{p,N,R}^{(\alpha)}(x) = {1 \over C_p^N}\int_R dt_1 \cdots  \int_R dt_N \,
\prod_{l=1}^N t_l^{\lambda_1} (1 - t_l)^{\lambda_2} (x - t_l)^\alpha \\
\times \prod_{1 \le j < k \le N} | t_k - t_j|^{2 \tau} 
e_p(x - t_1,\dots, x - t_N),
\end{multline}
where $R$ is a contour in the complex plane such that the integrand
vanishes at the endpoint, e.g.~$R=[0,1]$ or $R=[x,1]$ when the
parameters $\lambda_1, \lambda_2, \alpha$ are all positive.
The symbol $C_p^N$ denotes the binomial coefficient $N$ choose
$p$.
It is known that this  family of multiple integrals satisfies the
differential-difference system \cite{Fo93}, \cite[\S 4.6.4]{Fo10}
\begin{equation}\label{5.1a}
(N - p) E_p J_{p+1}(x) 
= (A_p x + B_p) J_p(x) - x(x-1) {d \over dx} J_p(x) + D_p x ( x - 1) J_{p-1}(x),
\end{equation}
valid for $p=0,1,\dots,N$.
Here we have abbreviated $J_{p,N,R}^{(\alpha)}(x) =: J_p(x)$, and the coefficients
are specified by
\begin{align*}
A_p & = (N-p) \Big ( \lambda_1 + \lambda_2 + 2 \tau  (N - p - 1) + 2(\alpha + 1) \Big ) \\
B_p & = (p-N)  \Big ( \lambda_1 + \alpha + 1 +  \tau (N - p - 1) \Big ) \\
D_p & = p \Big ( \tau (N-p) + \alpha + 1 \Big ) \\
E_p & = \lambda_1 + \lambda_2 + 1 +  \tau (2N - p - 2) + (\alpha + 1).
\end{align*}
This differential-difference system is in fact
equivalent to a  matrix differential equation \cite{FR12}; see below. The latter, in the special case $\tau = 1/2$, were
first isolated in the literature in \cite{Da72}; for recent developments see \cite{RF19,FK22,FK24}.

For $\tau$ a positive integer the absolute value signs in the product over differences in (\ref{5.1}) can be
removed. Suppose we choose $R$ to be the unit circle in the complex plane --- this is valid for
$\lambda_2 > 0$ since the integrand vanishes at the (single) endpoint $t_l=1$ --- and parametrise by
writing $t_l = e^{i \theta_l}$. Simple manipulation then gives that the integrand in (\ref{5.1}) transforms to
being proportional to
\begin{equation}\label{5.1b}
\prod_{l=1}^N e^{i \theta_l ( \lambda_1 + 1 +  \tau (N-1) + \alpha)} (1 - e^{i \theta_l})^{\lambda_2} (1 - x e^{- i \theta_l})^\alpha 
 \prod_{1 \le j < k \le N} | e^{i \theta_k} - e^{i \theta_j} |^{2 \tau} e_p(x - e^{i \theta_1}, \dots,
x - e^{i \theta_N}).
\end{equation}
As noted in \cite[Appendix A]{FK20}, direct use of integration by parts, starting with the integrand (\ref{5.1b})
in (\ref{5.1}), shows that the recurrence (\ref{5.1a}) in the case of (\ref{5.1b}) remains valid for general $\tau  > 0$.
We can now specify a differential-difference recurrence for a class of multiple integrals of the type appearing in
(\ref{eq:r2}).

\begin{prop}
Set
\begin{multline}\label{5.1c}
\tilde{J}_{p,N}(x) = {1 \over C_p^N}\int_{-\pi}^\pi d \theta_1 \cdots  \int_{-\pi}^\pi d \theta_N \,
\prod_{l=1}^N e^{i \tilde{\lambda}_1 \theta_l} (1 + e^{-i \theta_l})^{\lambda_2} e^{x e^{i \theta_l}} \\
\times \prod_{1 \le j < k \le N} |   e^{i \theta_k} - e^{i \theta_j}    |^{2 \tau} 
e_p( e^{-i \theta_1}    ,\dots,  e^{-i \theta_N}).
\end{multline}
Abbreviating $\tilde{J}_{p,N}(x) = \tilde{J}_{p}(x)$, for $p=0,1,\dots,N$
this multiple integral satisfies the differential-difference recurrence
\begin{equation}\label{5.1a+}
(N - p)   \tilde{E}_p  \tilde{J}_{p+1}(x) 
= (N - p) ( x  + \tilde{\lambda}_1 + p \tau)  \tilde{J}_p(x)  + x {d \over dx}  \tilde{J}_p(x)  - p x  \tilde{J}_{p-1}(x),
\end{equation}
where
\begin{equation}\label{E+}
\tilde{E}_p = - \tilde{\lambda}_1 + \lambda_2 + \tau(N - p - 1) + 1.
\end{equation}

\end{prop}

\begin{proof}
Up to a sign independent of $p$, the integrand in (\ref{5.1c}) follows from (\ref{5.1b}) by taking the complex
conjugate and setting
$$
 \theta_l \mapsto \theta_l - \pi, \quad - \tilde{\lambda}_1 = \lambda_1 + 1 + \tau(N - 1) + \alpha, \quad x \mapsto x/\alpha,  \quad \alpha \to \infty.
$$
Applying the last three of these to (\ref{5.1a}) gives (\ref{5.1a+}).

\end{proof}

Comparing (\ref{5.1c}) with (\ref{eq:r1}) in the case $\beta = 2$, we see that our immediate
interest in (\ref{5.1a+}) is for the choice of parameters  and substitutions
$$
N = 2, \quad \tau = 1, \quad x \mapsto - x, \quad \tilde{\lambda}_1 = a, \quad
\lambda_2 = N + a - 1, \quad \hat{J}_p(x) = x^a e^{-x} \tilde{J}_p(-x),
$$
when it reads
\begin{equation}\label{5.1a2}
(2 - p)  (N+1-p)  \hat{J}_{p+1}(x) 
=  ((1-p)(-x+a) + (2-p)p)  \hat{J}_p(x)  + x {d \over dx}  \hat{J}_p(x)  + p x  \hat{J}_{p-1}(x).
\end{equation}
This recurrence can be written in the equivalent matrix form
\begin{equation}\label{5.1a3}
x {d \over d x} \begin{bmatrix} \hat{J}_0(x) \\ \hat{J}_1(x)  \\  \hat{J}_2(x) \end{bmatrix} =
\begin{bmatrix} x - a & 2 (N + 1) & 0 \\
-x & - 1 & N \\
0 & - 2x  &  - x + a \end{bmatrix}
 \begin{bmatrix} \hat{J}_0(x) \\ \hat{J}_1(x)  \\ \hat{J}_2(x) \end{bmatrix}.
 \end{equation}
 The equation implied by the first row can be used to relate $\hat{J}_1(x)$ to $\hat{J}_0(x)$; 
 with this knowledge, the second equation allows $\hat{J}_2(x)$ to be written in
 terms of $\hat{J}_0(x)$. We can check that substituting in the equation implied by
 the third row gives that $\hat{J}_0(x)$ satisfies the third order scalar differential equation (\ref{u2}),
 as is consistent with the fact that $\hat{J}_0(x)$ is proportional to $\rho_{(1)}^{\rm (L)}(x)$.
 
 Now in (\ref{5.1a3}) replace $x$ by $s$ throughout. Multiply both sides by
 $e^{s/x}/s^{Na + N^2 - 1}$ $(0<x<1)$ and integrate with respect to $s$ over the contour
 $c - i \infty$ to $c + i \infty$. To do this, make use of integration by parts on the left hand
 side, and on the right hand side first decompose the matrix into a sum of one
 proportional to $s$, and one independent of $s$. After making use of  the analogue of
 (\ref{gg}) in the case $p=1$, this procedure gives a first order matrix differential equation for the column
 vector
 \begin{equation}\label{2.19}
 \bigg [ \int_{c - i \infty}^{c + i \infty}  {e^{s/x} \over s^{Na + N^2 - 1}} 
 \hat{J}_p(s) \, ds \bigg ]_{p=0,1,2}.
  \end{equation}
 Multiply (\ref{2.19}) by $x^{Na + N^2 - 1}$, and denote the resulting column vector by $[ \hat{I}_p(x) ]_{p=0,,1,2}$.
 The matrix differential equation for (\ref{2.19}) then transforms to read
 \begin{equation}\label{2.19a}
 x {d \over d x}  \begin{bmatrix} \hat{I}_0(x) \\ \hat{I}_1(x)  \\ \hat{I}_2(x) \end{bmatrix}  =
 \Big ( (N a + N^2 - 2) x - x^2 {d \over d x} \Big )
  \begin{bmatrix}  1 & 0 & 0 \\ - 1 & 0 & 0 \\ 0 & -2 & - 1
  \end{bmatrix}  \begin{bmatrix} \hat{I}_0(x) \\ \hat{I}_1(x)  \\ \hat{I}_2(x) \end{bmatrix} +
   \begin{bmatrix} -a & 2 (N+1) & 0 \\ 0 & -1 & N \\ 0 & 0 & a
   \end{bmatrix}   \begin{bmatrix} \hat{I}_0(x) \\ \hat{I}_1(x)  \\ \hat{I}_2(x) \end{bmatrix}.
   \end{equation}
   We now follow the  strategy detailed below (\ref{5.1a3}) to deduce a differential
   scalar equation for  $\hat{I}_0(x)$. This can be checked to be the same differential equation
   as (\ref{Aa1X}). Since by definition
   $\hat{I}_0(x)$ is proportional to $\rho_{(1)}^{\rm (fL)}(x)$ --- recall (\ref{0.3d}),
   the definition of $ \hat{I}_0(x)$ below
   (\ref{2.19}), and the fact that
   $\hat{J}_0(x)$ is proportional to $\rho_{(1)}^{\rm (L)}(x)$ --- this gives an alternative derivation of the latter.
   
   \subsection{A comparison of properties of $\rho_{(1)}^{\rm (L)}(x)$ and $\rho_{(1)}^{\rm (fL)}(x)$}\label{S2.3}
   Substituting (\ref{0.2}) in (\ref{u1}) shows that the density of the LUE has the structure
   \begin{equation}\label{ds1}
   \rho_{(1)}^{\rm (L)}(x) = x^a e^{-x} \sum_{p=0}^{2 (N - 1)} \alpha_p x^p
   \end{equation}
   for some coefficients $\{ \alpha_p \}$. As noted in \cite[Remark 2.1.3]{RF19} the differential
   equation (\ref{u2}) admits a unique solution with the $x \to \infty$ form compatible with
   (\ref{ds1})
 \begin{equation}\label{ds1a}
   \rho_{(1)}^{\rm (L)}(x) = { N C_{N-1,a}^{\rm (L)} \over C_{N,a}^{\rm (L)} } x^a e^{-x} 
   \Big ( x^{2 (N - 1)} + O(x^{2 (N - 1)-1}) \Big ),
   \end{equation}  
and  so in particular implies a recursion for $\{ \alpha_p \}$, with a unique solution
upon requiring  that
 \begin{equation}\label{ds1a+}
 \alpha_{2(N-1)} =   { N C_{N-1,a}^{\rm (L)} \over C_{N,a}^{\rm (L)} }.
 \end{equation}  
Substituting (\ref{ds1}) in (\ref{0.3d}) shows that in terms of $\{ \alpha_p \}$,
 \begin{equation}\label{ds1b}
   \rho_{(1)}^{\rm (fL)}(x) = {  C_{N,a}^{\rm (L)} \over C_{N,a}^{\rm (fL)} }
   \sum_{p=0}^{2 (N - 1)} \Gamma((N-1)a + N^2 - 2 - p + 1) \alpha_p
   \Big ( {1 \over x} - 1 \Big )^{(N-1)a + N^2 - 2 - p},
  \end{equation}   
  supported on $(0,1)$.  
  This indicates that the natural independent variable associated with (\ref{Aa1Y})
  is $y = (1/x)-1$, and that in this variable the differential equation admits
  a unique solution for $y \to \infty$ subject only to the specification of the
  leading term as made explicit by (\ref{ds1a+}).
  
  Consider next the moments associated with the densities
  $\rho_{(1)}^{\rm (L)}(x)$ and $\rho_{(1)}^{\rm (fL)}(x)$,
 \begin{equation}\label{ds1c}  
 m_p^{(\rm L)} = \int_0^\infty x^p \rho_{(1)}^{\rm (L)}(x) \, dx, \qquad
 m_p^{(\rm fL)} = \int_0^1 x^p \rho_{(1)}^{\rm (fL)}(x) \, dx.
 \end{equation}
 We know from \cite{HT03} and \cite{Le04} that $\{ m_p^{(\rm L)} \}$ satisfies the
 second order recurrence
  \begin{equation}\label{ds1d}    
  (k + 1) m_k^{(\rm L)} = (2k - 1) ( a + 2N ) m_{k-1}^{(\rm L)} +
  (k - 2) ( (k - 1)^2 - a^2 )  m_{k-2}^{(\rm L)} ,
  \end{equation}
  which was shown in \cite{CMSV16} to be valid for general $k > - a - 1$.
  One way to derive (\ref{ds1d}) is to multiply both sides of the
  differential equation (\ref{u2}) by $ x^{k-2}$, integrate both
  sides from $0$ to $\infty$, and simplify making use of integration by
  parts.  A formula for $m_k^{(\rm L)}$ in terms of certain continuous
  Hahn polynomials of degree $N-1$ and valid for continuous $k$
  has been given in \cite{CMOS18}.

Making use of the first equality in (\ref{0.3d}), as well as 
the explicit evaluations of the normalisation (\ref{0.2a}) and
(\ref{0.3a}), we can deduce \cite{ZS01}
 \begin{equation}\label{ds1e}    
 m_p^{(\rm fL)} = {1 \over (N^2+Na)_p}  m_p^{(\rm L)} ,
  \end{equation}
  where $(u)_p$ denotes the rising Pochhammer symbol.
  Substituting (\ref{ds1e}) in (\ref{ds1d}) gives that
  $\{ m_p^{(\rm fL)} \}$ satisfies the
 second order recurrence
  \begin{multline}\label{ds1eq}    
  (k + 1)(N^2 +  Na + k)(N^2 +  Na + k-1) m_k^{(\rm fL)}  \\
  = (2k - 1) ( a + 2N )(N^2 +  Na + k-1) m_{k-1}^{(\rm fL)} +
  (k - 2) ( (k - 1)^2 - a^2 )  m_{k-2}^{(\rm fL)}.
  \end{multline}
 Starting with the differential equation (\ref{Aa1Y}), and following
 the procedure noted below (\ref{ds1d}) gives an alternative derivation of
 (\ref{ds1eq}).
 
 We see from (\ref{0.2a}) that
 \begin{align*}
 C_{N,a}^{(\rm L)} & = \int_0^\infty d \lambda_1 \cdots  \int_0^\infty d \lambda_N \,
  \prod_{l=1}^N \lambda_l^a e^{- \lambda_l} \prod_{1 \le j < k \le N} (\lambda_k - \lambda_j)^2 \\
  & = t^{N^2 + Na} \int_0^\infty d \lambda_1 \cdots  \int_0^\infty d \lambda_N \,
  \prod_{l=1}^N \lambda_l^a e^{- t \lambda_l} \prod_{1 \le j < k \le N} (\lambda_k - \lambda_j)^2.
  \end{align*}
  Multiplying through by $1/( C_{N,a}^{(\rm L)}  t^{N^2 + N a})$, differentiating both sides with respect to $t$, and setting
  $t=1$ shows that
    \begin{equation}\label{ds1f}   
  m_1^{(\rm L)} = N^2 + N a.
  \end{equation}
  Since from the definition (\ref{0.3})
  \begin{equation}\label{ds1g}   
  m_1^{(\rm fL)} = 1
   \end{equation}
   we see that (\ref{ds1f}) is consistent with (\ref{ds1e}).
   
   The global scaling regime of positive definite matrices seeks a rescaling of
   the eigenvalues so that the first moment is of order $N$ as $N \to \infty$
   (as a normalisation we will require that the leading value be exactly $N$).
   Let us also scale the parameter $a$ by writing $a = \alpha N$. For the LUE, it
   follows from (\ref{ds1f}) that global scaling can  can be achieved by replacing each $\lambda_l$
   by $\tilde{\lambda}_l N $; to match with the first (scaled) first moment for the
   fLUE we should replace each $\lambda_l$
   by $\tilde{\lambda}_l /(N(1 + \alpha))  $. Thus in (\ref{Aa1Y}) we should make the 
   substitution  $x=yN$, and in (\ref{u2}) the substitution $x = y/(N(1 + \alpha))$. Doing this,
   and equating terms at leading order in $N$ gives in both cases the same
   first order linear differential equation
    \begin{equation}\label{d7}
    \Big ( -(  y^2 - 2( \alpha + 2) y + \alpha^2) y {d \over dy} + (\alpha + 2) y - \alpha^2 \Big )
    f(y) = 0.
    \end{equation}
    The solution of this, which is to be non-negative and normalised to integrate to unity,
    is the well known Marchenko-Pastur law
   \begin{equation}\label{d8}   
   f(y) = {1 \over 2 \pi y} \Big ((y - \alpha_-)(\alpha_+ - y) \Big )^{1/2}, \quad \alpha_\pm = ( \sqrt{\alpha+1} \pm 1)^2,
     \end{equation}
     supported on $y \in (\alpha_-,\alpha_+)$. This is consistent with results known from 
     \cite{Pa93,NMV10}.
     
     \section{Moments and cumulants for the purity and related statistics}
     \subsection{Cumulants of the purity and $\sigma$-Painlev\'e IV}\label{S3.1}
     We begin this section by recalling from \cite{FW00}, \cite[Ch.~8]{Fo10} the precise relationship
     between the GUE average in (\ref{G1}) and a transcendent in the Hamiltonian formulation of
     Painlev\'e IV \cite{Ok86}. Thus we have
  \begin{equation}\label{B1}
  {d \over dt} \log  \Big \langle    \prod_{j=1}^N( \lambda_j - t )^a \chi_{\lambda_j > t }  \Big \rangle^{\rm (G)}  = U_N(t,a),
  \end{equation}       
where $U_N(t,a),$ satisfies the particular $\sigma$-Painlev\'e IV equation, which is a second order
nonlinear equation satisfied by the Hamiltonian itself,
 \begin{equation}\label{B2}
 ( \sigma'')^2 - 4( t \sigma' - \sigma)^2 +4 \sigma' (\sigma' - 2 a) (\sigma' + 2 N) = 0.
  \end{equation}     
  Being a logarithmic derivative which equals a Hamiltonian in the Hamiltonian formulation of
  a Painlev\'e system, the average is said to be a $\tau$-function; see e.g.~\cite[\S 8.2.3]{Fo10}.
  
  In view of (\ref{B1}), writing in (\ref{G1}) $1/(2\sqrt{s}) = t$ we have that
   \begin{equation}\label{B3}
   {d \over d t} \log \hat{P}_{N,a}^{(\rm L)}(1/(4 t^2)) = 2 N t + {N^2 + N a \over t} + U_N(t,a).
   \end{equation}
   Being the exponential generating function for the distribution  of the linear
   statistic $\sum_{j=1}^N \lambda_j^2$ in the LUE, and with
   $\{\kappa_n^{(\rm L)}\}$ denoting the corresponding cumulants, we have the expansion
  \begin{equation}\label{B4}   
 \log \hat{P}_{N,a}^{(\rm L)}(u)  = \sum_{n=1}^\infty {(-1)^n \kappa_n^{(\rm L)} \over n!} u^n,
 \end{equation}
and consequently $U_N(t;a)$ permits the large-$t$ expansion
  \begin{equation}\label{B5}  
U_N(t,a) = -2Nt - {N^2 + N a \over t}  + {1 \over 2 t^3} \sum_{n=0}^\infty {(-1)^n \kappa_{n+1}^{(\rm L)} \over 4^{n} n! t^{2n}}.
 \end{equation}
 
 \begin{prop}\label{p3.1}
 The cumulants $\{  \kappa_n^{(\rm L)} \}_{n=1}^\infty $ for the distribution of the LUE linear statistic $\sum_{j=1}^N \lambda_j^2$ 
are proportional to the coefficients in the $\{1/t^{2n+3} \}$ expansion $(n=0,1,\dots)$ of the $\sigma$-Painlev\'e IV transcedent
 $U_N(t,a)$ according to  (\ref{B5}). These coefficients follow by substituting (\ref{B5}) in (\ref{B2}) which gives for the first
 three cumulants
 \begin{align*}
 \kappa_1^{\rm (L)} & = N (N + a) (2 N + a) \\
  \kappa_2^{\rm (L)} & = 2 N (N + a) (1 + 2 a^2 + 9 a N + 9 N^2) \\
  \kappa_3^{\rm (L)} & =  8 N ( N + a) ( 2N + a) ( 10 + 5 a^2 + 27 a N + 27 N^2).
  \end{align*}
  \end{prop}
  
  The cumulants fully determine the corresponding moments according to standard formulas
 \begin{equation}\label{cm}
   \hat{m}_1^{\rm (L)}  =  \kappa_1^{\rm (L)}, \quad   \hat{m}_2^{\rm (L)}  =   \kappa_2^{\rm (L)} + ( \kappa_1^{\rm (L)})^2, \quad
    \hat{m}_3^{\rm (L)}  =   \kappa_3^{\rm (L)} + 3  \kappa_2^{\rm (L)}  \kappa_1^{\rm (L)} + ( \kappa_1^{\rm (L)})^3.
\end{equation}
 The significance of knowledge of the moments $\hat{m}_n^{(\rm L)}$ is that they relate
 to the moments for the purity, $\hat{m}_n^{(\rm fL)}$ say, by a formula analogous to
 (\ref{ds1e}) \cite{VPO16}
\begin{equation}\label{ds1e+}    
 \hat{m}_p^{(\rm fL)} = {1 \over (N^2+Na)_{2p}}  m_p^{(\rm L)} .
  \end{equation}
Inverting the relations (\ref{cm}) then allows the cumulants $\kappa_n^{(\rm fL)}$ for the purity 
be computed.

\begin{cor}\label{C3.2}
The first three cumulants for the purity statistic $\sum_{j=1}^N \lambda_j^2$ associated with
the eigenvalue PDF (\ref{0.3}) are given by
\begin{align*}
 \kappa_1^{\rm (fL)} & =   {(2 N + a) \over (N(N+a) + 1)} \\
  \kappa_2^{\rm (fL)} & = {2 (N^2 - 1) ( (N+a)^2  - 1) \over  (N(N+a) + 1)^2 (N(N+a) + 2) (N(N+a) + 3) }\\
  \kappa_3^{\rm (fL)} & =  {8  (N^2 - 1) ( (N+a)^2  - 1) (2 N + a)  (N(N+a) - 5) \over
   (N(N+a) + 1)^3 (N(N+a) + 2) (N(N+a) + 3)  (N(N+a) + 4)  (N(N+a) + 5) }.
  \end{align*}
\end{cor}

The mean $ \kappa_1^{\rm (fL)}$ was first calculated in \cite{Lu78}. This was extended to the first
three cumulants in \cite{SC03}, and, using the sum formula (\ref{d1b}) for the moments, to the first
five cumulants in \cite{Gi07}. Although we have listed only the first three cumulants above, our
method also provides an efficient and systematic way to continue this list.

\subsection{The linear statistic $\sum_{k=1}^N \lambda_j^q$}\label{S3.2}
The linear statistic $T_q := \sum_{k=1}^N \lambda_j^q$ is closely related to the quantum Tsallis
entropy introduced in the paragraph including (\ref{Aa}), with the latter specified in terms of $T_q$ by
$\tilde{T}_q := (1/(1-q)) (1 - T_q)$. Here we will present the
analogue of the summation expression (\ref{d1b}) for the moments of $T_q$. Note that these
relate to the moments of $\tilde{T}_q$ by the simple formula
\begin{equation}\label{Lw}
( \tilde{T}_q)^k = \Big ( {1 \over 1 - q} \Big )^k \sum_{s=0}^k (-1)^s \binom{k}{s} (T_q)^s.
\end{equation}

 \begin{prop}\label{P3.2}
For all $r \ge 0$, $k \in \mathbb Z_{\ge 0}$ and $qk+a > - 1$ we have
\begin{multline}\label{d1bqs}
\Big \langle \Big (   \sum_{j=1}^N \lambda_j^q \Big )^k \Big \rangle^{\rm (fL)} =
{(N (N + a) - 1)! \over ( N(N + a) + qk - 1)!}  \\ \times \sum_{k_1,\dots,k_N \ge 0 \atop
k_1 + \cdots + k_N = k}
{k! \over k_1! \cdots k_N!} \prod_{i=1}^N
{(N + a + qk_i - i)! \over (N + a - i)! }
\prod_{1 \le i < j \le N} { q k_i - i - qk_j + j \over -i + j }.
\end{multline}
\end{prop}

\begin{proof}
The working in \cite{Gi07} for the case $q=2$ is sufficient to derive the general $q$ result (\ref{d1bqs}).

With
\begin{equation}\label{J1}
J_{a,N}(\mathbf k) := N! \int_0^1 dx_1 \cdots  \int_0^1 dx_N \,  \delta \Big ( 1 - \sum_{i=1}^N x_i \Big ) \Delta(\mathbf x) 
\prod_{i=1}^N x_i^{a + k_i + i - 1},
\end{equation}
where $\mathbf k = (k_1,\dots, k_N)$, $\Delta(\mathbf x) = \prod_{1 \le i < j \le N} (x_j - x_i)$, this working
relied on the evaluation formula
\begin{equation}\label{J2}
J_{a,N}(\mathbf k) =  {N!    \prod_{i=1}^N  \Gamma( a + k_i + i ) \over  \Gamma(N^2 + a N + \sum_{i=1}^N k_i - 1)!}
\prod_{1 \le i < j \le N} (k_j - k_i + j - i).
\end{equation}
To deduce (\ref{J2}),
introduce the Schur polynomial $s_{\kappa_1,\dots,\kappa_N}(\mathbf x)$ specified
by
\begin{equation}\label{sf}
s_{\kappa_1,\dots,\kappa_N}(\mathbf x) = \det[ x_j^{\kappa_{N+1 - i} + i -1} ]_{i,j=1}^N/ \Delta (\mathbf x);
\end{equation}
conventionally it is required that $\kappa_1 \ge \cdots \ge \kappa_N \ge 0$ form a partition of
non-negative integers, however (\ref{sf}) is well defined without this assumption although it will no longer necessarily
be a polynomial.
By symmetrising the integrand of (\ref{J1}) the final factor therein becomes a determinant, and making use of
(\ref{sf}) shows
\begin{equation}\label{J3}
J_{a,N}(\mathbf k) :=  \int_0^1 dx_1 \cdots  \int_0^1 dx_N \,  \delta \Big (1 -  \sum_{i=1}^N x_i \Big ) (\Delta(\mathbf x) )^2
s_{a+k_N,\dots,a+k_1}(\mathbf x).
\end{equation}
On the other hand, from the theory of Schur polynomials it is well known \cite{Hu63} (see \cite[\S 2.3]{FI16} for
another recent application in random matrix theory) that
\begin{equation}\label{J4}
 \int_0^\infty dx_1 \cdots  \int_0^\infty dx_N \, \prod_{l=1}^N x_l^a  e^{- x_l}  (\Delta(\mathbf x) )^2
s_{\kappa_1,\dots, \kappa_N}(\mathbf x) = N!  \prod_{i=1}^N  \Gamma( a + k_i + i ) \prod_{1 \le i < j \le N} (\kappa_i - \kappa_j + j - i).
\end{equation}
By using the same strategy as in the derivation of (\ref{ds1e}), we can deduce from this the evaluation (\ref{J2}).

To make use of (\ref{J2}) in the derivation of (\ref{d1bqs}), following the working in
 \cite{Gi07} for the case $q=2$ we begin by noting from the antisymmetry of $\Delta(\mathbf x)$ that
 \begin{equation}\label{J5}
\Big \langle \Big (   \sum_{j=1}^N \lambda_j^q \Big )^k \Big \rangle^{\rm (fLUE)} =
{N! \over C_{N,a}^{\rm (fL)}}  \int_0^1 dx_1 \cdots  \int_0^1 dx_N \,  \delta \Big ( 1 - \sum_{i=1}^N x_i \Big ) \Delta(\mathbf x) 
\prod_{i=1}^N x_i^{a  + i - 1} \Big ( \sum_{j=1}^N x_j^q \Big )^k.
 \end{equation}
Now substituting 
\begin{equation}\label{J6}
\Big ( \sum_{j=1}^N x_j^q \Big )^k = \sum_{k_1,\dots, k_N \ge 0 \atop k_1 + \cdots + k_N = k} {k!\over k_1! \cdots k_N!}
x_1^{q k_1} \cdots x_N^{q k_N} ,
\end{equation}
and taking the sum outside the integral, allows each of the integrals
to be evaluated according to (\ref{J2}), and (\ref{d1bqs}) follows.
\end{proof}

In the case $k=1$ the multi-index summation in (\ref{d1bqs}) can be replaced by a sum over $j$, $1 \le j \le N$, where
$k_j = 1$, $k_p= 0$ $(p \ne s$). This allows for the case $k=1$ to be expressed in terms of a terminating
${}_3 F_2$ hypergeometric function of unit argument.

\begin{prop}\label{P3.4}
We have
\begin{multline}\label{d1bq}
\Big \langle    \sum_{j=1}^N \lambda_j^q   \Big \rangle^{\rm (fL)} =
{(N (N + a) - 1)! \over ( N(N + a) + q - 1)!} 
{\Gamma(N+a+q) \over \Gamma(N+a)} {\Gamma(N+q) \over \Gamma(1 + q) \Gamma(N)} \\
\times {}_3 F_2 \bigg ( {1-N, 1-(N+a), 1 - q \atop 1 - (N-q), 1 - (N+a+q)} \bigg | 1 \bigg ).
\end{multline}
\end{prop}

\begin{proof}
After parameterising the multi-index summation as indicated,  the single product in the sum reduces to
$$
{(N+a+q - j)! \over (N+a - j)! } = {\Gamma(N + a + q) \over \Gamma(N+a)} {(1 - (N + a))_{j-1} \over (1 - (N + a+q))_{j-1} },
$$
and the double product reduces to
$$
\prod_{i=1}^{j-1} { - i - q + j \over -i + j } \prod_{i=j+1}^N {i + q - j \over i - j} =
{(1 - q)_{j-1} \over (j-1)!} {\Gamma(N + q) \over \Gamma(1 + q) \Gamma(N)} {(1 - N)_{j-1} \over (1 - (N + q))_{j-1}}.
$$
The summation over $j$ can be recognised as the series form of a particular ${}_3 F_2$ function,
implying (\ref{d1bq}).

\end{proof}

Note that for $q$ a positive integer less than $N$ the hypergeometric summation in
(\ref{d1bq}) has only $q$ nonzero terms. In particular, for $q=1$ we can check that
the result for $\kappa_1^{(\rm fL)}$ in Corollary \ref{C3.2} is recovered.  

Using a different formalism to our Proposition \ref{P3.4}, Wei in the work \cite{We19}\footnote{For $q,a$ non-negative integers, this in an equivalent form can be found in \cite[Th.~2.5]{HSS92}.} 
has derived the formula
\begin{equation}\label{d1bq+}
\Big \langle    \sum_{j=1}^N \lambda_j^q   \Big \rangle^{\rm (L)} =
{N \Gamma(N+a+q) \over \Gamma(N+a)}
\, {}_3 F_2 \bigg ( {1-N, -q ,1 - q \atop  2, 1 - (N+a+q)} \bigg | 1 \bigg ).
\end{equation}
In keeping with (\ref{ds1e})
$$
\Big \langle    \sum_{j=1}^N \lambda_j^q   \Big \rangle^{\rm (fL)} =
{(N (N + a) - 1)! \over ( N(N + a) + q - 1)!} \Big \langle    \sum_{j=1}^N \lambda_j^q   \Big \rangle^{\rm (L)}
$$
so after substituting (\ref{d1bq}) on the left hand side, and (\ref{d1bq+}) on the right hand side,
it follows that we must have
\begin{equation}\label{d1bq+1}
 {\Gamma(N+q) \over \Gamma(1 + q) }
\, {}_3 F_2 \bigg ( {1-N, 1-(N+a), 1 - q \atop 1 - (N-q), 1 - (N+a+q)} \bigg | 1 \bigg )
= N! 
\, {}_3 F_2 \bigg ( {1-N, -q, 1 - q \atop  2, 1 - (N+a+q)} \bigg | 1 \bigg ).
\end{equation}
This can be recognised as a special case of the more general identity \cite[Entry (7.4.4.85)]{PBM90}
\begin{equation}\label{d1bq+s}
\, {}_3 F_2 \bigg ( {-n, b, c  \atop d, e} \bigg | 1 \bigg )
= {(d - b)_n \over (d)_n} 
\, {}_3 F_2 \bigg ( {-n , b , e-c \atop  e, b-d-n+1} \bigg | 1 \bigg ),
\end{equation}
valid for non-negative integer $n$.

In \cite{Ok21}, Okuyama used the replica method to deduce (\ref{d1bq}), and noted too the alternative form
provided by using (\ref{d1bq+1}). The latter was then used to obtain the further rewrite of (\ref{d1bq})
\begin{equation}\label{d1bqA}
\Big \langle    \sum_{j=1}^N \lambda_j^q   \Big \rangle^{\rm (fL)} =
{\Gamma(N) \Gamma(N(N+a)+1) \over \Gamma(a+N+1) \Gamma(N(N+a)+q)}
\sum_{k=1}^N {\sf N}_{q,k} {\Gamma(a+N+1+q+k) \over \Gamma(N+1-q)},
\end{equation}
where ${\sf N}_{q,k} := {1 \over q} \binom{q}{k} \binom{q}{k-1}$ are the Narayama numbers. From this form it was demonstrated
that
\begin{equation}\label{d1bqB}
\lim_{N,a \to \infty \atop N/(N+a) \to \alpha} N^{q-1} \Big \langle    \sum_{j=1}^N \lambda_j^q   \Big \rangle^{\rm (fL)} = \sum_{k=1}^\infty {\sf N}_{q,k} 
\alpha^{k-1},
\end{equation}
thus reclaiming a result from \cite{KF21} derived using a different approach. Note that for $q$ a positive integer, the sum on the right hand side is finite, with upper terminal $q$, and it specifies the integer moments of the Marchenko-Pastur law (\ref{d8}).

The work \cite{We19} also derived a formula for $\langle ( \tilde{T}_q)^2 \rangle^{(L)}$, involving
finite single and double sums over particular ${}_3 F_2$ functions. Specialising to $q=2$ gives a formula
\cite[Eq.~(58)]{We19} consistent with the moment $\hat{m}_2^{(L)}$
as implied by Proposition \ref{p3.1} and (\ref{cm}). Moreover, as initiated in \cite{BD19,SK19}, it was
shown in \cite{We19} how knowledge of $\langle \tilde{T}_q \rangle^{(\rm fL)}$ and $\langle (\tilde{T}_q)^2 \rangle^{(\rm fL)}$,
as implied by first computing  $\langle {T}_q \rangle^{(\rm fL)}$ and $\langle ({T}_q)^2 \rangle^{(\rm fL)}$,
then using (\ref{Lw}), can be used to facilitate the calculation of the variance for von Neumann entropy
\cite{VPO16, We17}. For recent related work on this theme see \cite{We20a,We20b,WW21,BHK21,HWC21,LW21}.

We also remark that the moments $\Big \langle    \sum_{j=1}^N \lambda_j^q   \Big \rangle^{\rm (L)}$ for $q$ continuous have
featured in the recent study \cite{CMOS18}, where they are specified in 
terms of a certain continuous
  Hahn polynomial of degree $N-1$, as noted below (\ref{ds1d}).
  The latter are known to be given in terms of a terminating ${}_3 F_2$ hypergeometric function,
  and in fact such a formula was presented in \cite[Eq.~(4.11)]{CMOS18}. This can also be shown to be
  equivalent to (\ref{d1bq+})  upon appropriate use of (\ref{d1bq+s}).

  \begin{remark} ${}$ \\
  1.~Generally in random matrix theory, the average of the linear statistic $\sum_{j=1}^N \lambda_j^q$ correspond to the moments of the
  spectral density, and have been studied for many decades now. 
  Motivating these studies has been the analysis of the global large $N$ limit \cite{Wi58}, their combinatorial
  significance \cite{BIPZ78,HZ86}, applications to fluctuation formulas (ses e.g.~the review \cite{Fo23}), as well as  their integrability properties as
  further developed here. Additional recent works on the latter theme not already cited include \cite{By24,BF24,BFO24,ABO25,BJO25}. \\
  2.~As with the case $q=2$, Painlev\'e structures are also present for $q=-1$. For this value of $q$, from the condition $qk+a > - 1$ in
  Proposition \ref{P3.2}, note that the left hand side of (\ref{d1bqs}) is only defined for $k < a + 1$. On the other hand, the Laplace-Fourier transform of the linear
  statistic $\sum_{j=1}^N \lambda_j^{-1}$,
  \begin{equation}\label{d1Lz}
\hat{Q}_{N,a}^{\rm (L)}(s) := {1 \over C_{N,a}^{(\rm L)}}  \int_0^\infty d\lambda_1 \cdots \int_0^\infty d\lambda_N \,  
\prod_{l=1}^N e^{-\lambda_l -{s  \over  \lambda}}
 \lambda_l^a \prod_{1 \le j < k \le N} (\lambda_k - \lambda_j)^2,
\end{equation}
is well defined for all Re$(s) \ge  0$ (although is not analytic at $s=0$).  It has been shown in
multiple works \cite{OK07,CI10,MS13,DFX22} from various viewpoints that  $s {d \over ds} \log \hat{Q}_{N,a}^{\rm (L)}(s)$
is a particular  $\sigma$-Painlev\'e PIII$'$ trascendent. An application given in \cite{XDZ14,DFX22} shows that $\hat{Q}_{N,a}^{\rm (L)}(s/N)$,
and thus the cumulants $\kappa_k/N^k$ ($k<a+1$), have a well defined $N \to \infty$ limit, with the former given as the solution
of a particular third order nonlinear equation. 
\end{remark}

  \section{Purity moments for general $\beta$}\label{S4}
  The eigenvalue PDFs (\ref{0.2}) and (\ref{0.3}) admit natural $\beta$ generalisations, 
  \begin{equation}\label{0.2beta}
p_{N,a,\beta}^{(\rm L{}_\beta)}(\lambda_1,\dots,\lambda_N) := {1 \over C_{N,a,\beta}^{(\rm L{}_\beta)}} \prod_{l=1}^N \lambda_l^a e^{- \lambda_l} \prod_{1 \le j < k \le N} |\lambda_k - \lambda_j|^\beta \end{equation}
(this is (\ref{0.2a+}) up to scaling of the $\lambda_l$) and
\begin{equation}\label{0.3beta}
p_{N,a,\beta}^{(\rm fL{}_\beta)}(\lambda_1,\dots,\lambda_N) :=
{1 \over C_{N,a,\beta}^{(\rm fL{}_\beta)}}  \delta \Big (1 -  \sum_{l=1}^N \lambda_l \Big ) \prod_{l=1}^N \lambda_l^a \prod_{1 \le j < k \le N} |\lambda_k - \lambda_j|^\beta,
\end{equation}
which are well established in random matrix theory \cite{Fo10}. The parameter $\beta$ is often referred to as the Dyson index. In this section we will
take up the task of computing the positive integer moments of the  statistic $T_2 = \sum_{j=1}^N \lambda_j^2$ 
(i.e.~the purity in the context of quantum information)
with respect to (\ref{0.2beta}).
The same moments with respect to (\ref{0.3beta}) then follow from the analogue of (\ref{ds1e}),
\begin{equation}\label{ds1k}    
 m_p^{(\rm fL{}_\beta)} = {1 \over (\beta N (N - 1)/2 +N(a+1))_{2p}}  m_p^{(\rm L{}_\beta)}.
  \end{equation}
An explicit formula for the mean $ m_1^{(\rm L{}_\beta)}$ can be read off from a result
in \cite[Eq.~(A.9b)]{MRW15}
\begin{align}\label{mrs}
 m_1^{(\rm L{}_\beta)} & = N \Big ( a^2 + 3a - {4  \tau} - {3 a  \tau}
 +2 + {2  \tau^2} \Big ) + N^2 \Big (  {3 a  \tau} - {4  \tau^2} + {4  \tau} \Big )
 + {2 N^3  \tau^2} \nonumber \\
 & =  N  \Big ( \tau (N - 1) + (1+a) \Big ) \Big (2 \tau ( N - 1)  +(2+a)  \Big ), 
 \end{align}
 where $\tau = \beta / 2$. Notice that for $\beta = 2$ this reduces to 
 the result for $\kappa_1^{\rm (L)}$ given in Proposition \ref{p3.1}.
 
 The positive integer moments for general $\beta$ can be studied by considering the Laplace-Fourier
 transform
 \begin{equation}\label{d1Lb}
\hat{P}_{N,a,\beta}^{\rm (L{}_\beta)}(s) = {1 \over C_{N,a,\beta}^{(\rm L{}_\beta)}}  \int_0^\infty d\lambda_1 \cdots \int_0^\infty d\lambda_N \,  
\prod_{l=1}^N e^{-\lambda_l -s   \lambda_l^2}
 \lambda_l^a \prod_{1 \le j < k \le N} |\lambda_k - \lambda_j|^\beta
\end{equation}
(cf.~(\ref{d1L})). With the $\beta$ generalisation of the GUE eigenvalue PDF (\ref{0.1}) specified by 
 \begin{equation}\label{0.1b}
p_{N,\beta}^{(\rm G{}_\beta)}(\lambda_1,\dots,\lambda_N) := {1 \over C_{N,\beta}^{(\rm G{}_\beta)}} \prod_{l=1}^Ne^{- \lambda_l^2} \prod_{1 \le j < k \le N} |\lambda_k - \lambda_j|^\beta, \end{equation}
our analysis proceeds by first changing variables $\lambda_l \mapsto \lambda_l/\sqrt{s}$ in (\ref{d1Lb}) and completing the
square in the exponent to deduce the $\beta$ generalisation of (\ref{G1})
\begin{equation}\label{G1b}
\hat{P}_{N,a,\beta}^{\rm (L{}_\beta)}(s) = {C_{N,\beta}^{\rm (G{}_\beta)} \over  C_{N,a,\beta}^{\rm (L{}_\beta)}} e^{N/4s} (1 / \sqrt{s} )^{N  (a + 1)   + \beta N (N - 1)/2}
\Big \langle    \prod_{j=1}^N( \lambda_j - 1/(2 \sqrt{s}) )^a \chi_{\lambda_j > 1/(2 \sqrt{s})}  \Big \rangle^{\rm (G{}_\beta)}.
\end{equation}

The utility of (\ref{G1b}) stems from the fact that a generalisation of the multidimensional integral corresponding to the right hand side,
\begin{multline}\label{5.1b+}
I_{p,N,\tau}^{(\alpha)}(x) = {(-1)^p \over C_p^N}\int_x^\infty dt_1 \cdots  \int_x^\infty dt_N \,
\prod_{l=1}^N  e^{- t_l^2} ( t_l - x)^\alpha \\
\times \prod_{1 \le j < k \le N} | t_k - t_j|^{2 \tau} 
e_p( t_1 - x,\dots,  t_N - x),
\end{multline}
(cf.~(\ref{5.1b})) satisfies the differential-difference recurrence \cite[for the case $\tau = 1/2$]{Ek74}, \cite{FT19}
 \begin{equation}\label{5.1c+}
 (N-p) I_{p+1}(x) = (N - p) x  I_{p}(x) + {1 \over 2} {d \over d x} I_{p}(x)  - { p (\tau (N - p) + \alpha + 1) \over 2} I_{p-1}(x),
\end{equation} 
valid for $p=0,\dots,N$, where we have abbreviated $I_{p,N}^{(\alpha)}(x) =  I_{p}(x)$.
We can use (\ref{5.1c+}) to deduce a matrix differential equation (recall Section \ref{S2.2})
for a family of multiple integrals
generalising (\ref{d1Lb})
\begin{multline}\label{5.1d}
H_p(s) = H_{p,N,a,\tau}(s) =
{(-1)^p \over C_p^N}\int_0^\infty dt_1 \cdots  \int_0^\infty dt_N \,
\prod_{l=1}^N t_l^a e^{- t_l}  e^{- s t_l^2} \\
\times \prod_{1 \le j < k \le N} | t_k - t_j|^{2 \tau} 
e_p( t_1,\dots,  t_N).
\end{multline}

\begin{prop}
Introduce the column vector $\mathbf H(s) = [H_p(s) ]_{p=0}^N$, and the bidiagonal matrices
\begin{align}\label{hAB}
{\sf A} & = - {\rm diag} \, \Big ( N (a + 1) + \tau N (N - 1) + p \Big )_{p=0}^N - {\rm diag}{}^+  \, \Big ( N - p \Big )_{p=0}^{N-1} \nonumber \\
{\sf B} & = - {\rm diag} \,  \Big ( {p \over 2} \Big )_{p=0}^N -  {\rm diag}{}^-  \, \Big ({ p( \tau (N - p) + a + 1) \over 2} \Big )_{p=1}^{N}.
\end{align}
Here $ {\rm diag}{}^+$ ($ {\rm diag}{}^-$)    refers to the diagonal directly above (below) the main diagonal. 
We have
\begin{equation}\label{hAS1}
2 s^2 {d \over d s} \mathbf H(s)  = s {\sf A} \mathbf H(s) + {\sf B} \mathbf H(s).
\end{equation}
Also, denote by ${\sf B}_q$ the matrix
which has the first row in $\sf B$, labelled by $p=0$ and which consists of all zeros, replaced by the first row in $-{\sf A} + 2 q {\sf I}_{N+1}$.
Introducing the power series expansion
\begin{equation}\label{hAS2}
\mathbf H(s) = \sum_{q=0}^\infty \mathbf h_q s^q,
\end{equation}
we have the recurrence
\begin{equation}\label{hAS3}
 \mathbf h_{q+1} = B_{q+1}^{-1} ( - {\sf A} + 2 q {\sf I}_{N+1})  \mathbf h_{q},
\end{equation}
subject to the requirement that
\begin{equation}\label{hAS4}
{\sf B}  \mathbf h_0 = \mathbf 0.
\end{equation}
\end{prop}

\begin{proof}
From the definitions
\begin{equation}\label{hAS5}
 I_{p,N}^{(a)}(1/(2 \sqrt{s})) = (\sqrt{s})^{N(a+1) + \tau N (N - 1) + p} e^{- N/ 4 s} H_p(s).
 \end{equation}
 Thus after changing variables $x = 1/(2 \sqrt{s})$ in (\ref{5.1c+}), then substituting according
 to (\ref{hAS5}), we have that $\{ H_p(s) \}_{p=0}^N$ satisfies the differential-difference equation
 \begin{multline*}
 (N - p) s H_{p+1}(s) = (N - p) {1 \over 2} H_p(s) - \Big ( s ( N (a + 1) + \tau N (N - 1) + p) + {N \over 2} \Big ) H_p(s) - 2 s^2 {d \over d s} H_p(s) \\
 - {p ( \tau (N - p) + a + 1) \over 2} H_{p-1}(s),
 \end{multline*}
 valid for $p=0,\dots,N$. Writing this in matrix form gives (\ref{hAS1}).
 
 Substituting (\ref{hAS2}) in (\ref{hAS1}) shows
 \begin{equation}\label{hAS6}
 \sum_{q=0}^\infty ( - {\sf A} + 2 q {\sf I}_{N+1} ) \mathbf h_q s^{q+1} = \sum_{q=0}^\infty {\sf B} \mathbf h_q s^q.
 \end{equation}
 Equating the term independent of $s$ on both sides gives (\ref{hAS4}). Equating powers of
 $s^{q+1}$ for $q \ge 0$ gives
  \begin{equation}\label{hAS7}
  (- {\sf A} + 2 q {\sf I}_{N+1}) \mathbf h_q  = {\sf B}  \mathbf h_{q+1}.
 \end{equation}  
 The first row in $\sf B$, labelled by $p=0$, has all entries equal to zero. Thus the first entry in
 $ (- {\sf A} + 2 q {\sf I}_{N+1}) $ is equal to zero. This implies that the matrix $\sf B$ as it appears in (\ref{hAS7})
 can have its first row replaced by the first row in $-{\sf A} + 2 (q+1) {\sf I}_{N+1}$. The resulting matrix, denoted
 ${\sf B}_{q+1}$, is invertible and (\ref{hAS3}) follows.
 \end{proof}
  
 Writing $\mathbf h_q = (h_q^{(p)} )_{p=0}^N$ we see from (\ref{hAS2}), (\ref{d1Lb}) and (\ref{5.1d})
 that $h_0^{(0)} =  C_{N,a,\beta}^{(\rm L{}_\beta)}$. Moreover, it follows from (\ref{hAS4}) that
  \begin{equation}\label{hAS8}
  h_0^{(p)} = h_0^{(0)}  (-1)^p  \prod_{l=1}^p ( \tau (N - l ) + a + 1).
  \end{equation}  
  As a check, we note from the definitions that $h_0^{(N)} = (-1)^N  C_{N,a+1,\beta}^{(\rm L{}_\beta)}$. This
  substituted in (\ref{hAS8}) implies
   \begin{equation}\label{hAS9} 
   { C_{N,a+1,\beta}^{(\rm L{}_\beta)} \over  C_{N,a,\beta}^{(\rm L{}_\beta)}}  = \prod_{l=1}^N ( \tau(N - l) + a + 1),
\end{equation}  
which is indeed true as     $C_{N,a,\beta}^{(\rm L{}_\beta)}$ has the product of gamma function evaluation
(see e.g.~\cite[Prop.~4.7.3]{Fo10})
  \begin{equation}\label{hAS9a} 
C_{N,a,\beta}^{(\rm L{}_\beta)} = \prod_{j=0}^{N-1} { \Gamma(1 + (j+1) \tau) \Gamma(a + 1 + j \tau) \over
\Gamma(1 + \tau)}.
\end{equation}   
 
 Denote the first entry in $\mathbf h_q$ by $h_q^{(0)}$ so that
   \begin{equation}\label{X1}
   H_0(s) = h_0^{(0)}\sum_{q=0}^\infty {(-1)^q m_q^{(L_\beta)} \over q!} s^q = \sum_{q=0}^\infty h_q^{(0)} s^q.
 \end{equation}  
 Here the first equality follows by comparing   (\ref{5.1d})  and (\ref{d1Lb}), and performing
 a power series expansion in $s$ of the latter, recalling too that $m_q^{(\rm L_\beta)}$ refers
 to the moments of the linear statistic $T_2 = \sum_{j=1}^N \lambda_j^2$ with respect to (\ref{0.2beta}).
 Thus, in light of (\ref{hAS3}) and (\ref{hAS8}), for a given $N$ we have a computational scheme to evaluate
 $\{ m_q^{(\rm L_\beta)} \}_{q=1,2,\dots}$ which involves only multiplication of matrices of size $(N+1) \times (N + 1)$.
 Moreover, we can use the fact that each $m_q^{(L_\beta)}$ is a polynomial in $N$ of degree $3q$ (i.e. the
 same degree as $(m_1^{(L_\beta)})^q$) to deduce the dependence on $N$ from a table of $m_q^{(\rm L_\beta)}$
 (or more) different values of $N$. Proceeding in this way, extending the result (\ref{mrs}) for $m_1^{(L_\beta)}$
 we can deduce the explicit form of $m_2^{(L_\beta)}$. This simplifies upon introducing the corresponding
 cumulant $\kappa_2^{(\rm L_\beta)} = m_2^{(\rm L_\beta)} - (m_1^{(\rm L_\beta)})^2$.

 \begin{prop}
 Let $m_1^{(L_\beta)}$ be given by (\ref{mrs}), where we recall $\tau = \beta / 2$.
 We have
 \begin{equation}\label{k2+}
 m_2^{(\rm L_\beta)} = \kappa_2^{(\rm L_\beta)} + (m_1^{(\rm L_\beta)})^2,
 \end{equation}
 where
  \begin{equation}\label{k2}
 \kappa_2^{(\rm L_\beta)} =  2 N  \Big (1 + a + \tau (-1 + N) \Big ) \Big (10 + 2 a^2 + 9 a (1 +  \tau(-1 + N)) + 
    \tau (-1 + N) (19 + \tau (-10 + 9 N)) \Big ).
    \end{equation}
    \end{prop}

    \begin{remark}
    For general non-negative integer $q$, the moments $m_q^{(L_\beta)}$ are unchanged,
    up to a factor of $(-\tau)^q$, by the mapping of parameters \cite{DE06}
    \begin{equation}\label{nta}
    (N, \tau, a) \mapsto (-\tau N, 1/\tau, -a/\tau).
    \end{equation}
    Note that the explicit formula (\ref{mrs}) exhibits this invariance.
    Recently \cite{Fo21} it has been established that this property carries
    over to all cumulants, as indeed is exhibited  by (\ref{k2}) as one example.
    \end{remark}
    
    \begin{cor}
 For the fixed trace ensemble {\rm fL${}_\beta$}, the cumulant $\kappa_2^{(\rm fL_\beta)}$ is given by
 \begin{multline}\label{nta1}  
 \kappa_2^{(\rm fL_\beta)} = {2 \Big (10 + 2 a^2 + 9 a (1 +  \tau(-1 + N)) + 
    \tau (-1 + N) (19 + \tau (-10 + 9 N)) \Big ) \over
   (\tau N (N - 1) + N(a+1) + 1)_3} 
  \\ + { N (  \tau (N - 1) + a + 1) ( 2 \tau (N - 1) + a + 2)^2 \over
    (\tau N (N - 1) + N(a+1) + 1)_3} 
    -  \Big ( {2 \tau (N - 1) + a + 2 \over \tau N (N - 1) + N(a+1) + 1} \Big )^2.
   \end{multline}
    \end{cor}
    
    \begin{proof}
    It follows from (\ref{ds1k}) that
    $$
  \kappa_2^{(\rm fL_\beta)} = {m_2^{(\rm L_\beta)} \over      (\tau N (N - 1) + N(a+1) )_4}  -
\Big (  {m_1^{(\rm L_\beta)} \over (\tau N (N - 1) + N(a+1) )_2 } \Big )^2.
$$
Substituting (\ref{mrs})  and (\ref{k2+}) as appropriate gives (\ref{nta1}).
\end{proof}

\section*{Acknowledgements}
 This work of Peter J.~Forrester has been supported by the Australian Research Council
discovery project grant DP250102552. 
The work of S.M.~Nishigaki has been supported by Japan Society for the Promotion of Sciences
(JSPS) Grants-in-Aids for Scientific Research (C) No.\,7K05416,
and he gratefully acknowledges the financial support 
that made his visit to the University of Melbourne possible.
Sung-Soo Byun is to be thanked for providing helpful feedback on an earlier draft.

\bigskip
\section*{Appendix A}
\renewcommand{\thesection}{A} 
\setcounter{equation}{0}
\setcounter{prop}{0}

In the historical development of random matrix theory, the introduction of random unitary matrices preceded that of
random Hermitian matrices --- see the account in \cite{DF17}.
The natural analogue of a fixed trace Hermitian matrix for a random unitary matrix is to impose a unit determinant constraint.
Then, for example, the matrix group U$(N)$ becomes the subgroup SU$(N)$. 
With motivations in lattice gauge theories \cite{HW24},
the recent work by one of the present authors \cite{Ni24} took up the question of computing the exact one-point correlation
function for SU$(N)$, as well as the unit determinant versions of Dyson's circular unitary and symplectic ensembles; for an introduction 
to the latter see \cite[Ch.~2]{Fo10}.
In the case of SU$(N)$ it was found, by the use of the Selberg integral, that
\begin{equation}\label{A1}
\rho_{(1)}^{{\rm SU}(N)}(\theta) = {N \over 2 \pi} \Big ( 1 - (-1)^N{2 \over N} \cos N \theta \Big ).
\end{equation}
As an application, consider the linear statistic Tr$(U^p) = \sum_{j=1}^N e^{ip \theta_j}$, $p \in \mathbb Z$.
It follows by integrating this against (\ref{A1}) that
\begin{equation}\label{A2}
\Big \langle {\rm Tr}(U^p) \Big \rangle^{{\rm SU}(N)} =  \begin{cases} N, \quad p=0 \\ (-1)^{N-1}, \quad p= \pm N \\ 0, \quad {\rm otherwise}; \end{cases}
\end{equation}
for specific interest in this see \cite{PSt21}. Note that the corresponding average over U$(N)$ gives zero in all cases except $p=0$.

At the next level of complexity is the two-point correlation. 
In the case of U$(N)$ one has the classical result \cite{Me04}
\begin{equation}\label{A3}
\rho_{(2)}^{{\rm U}(N)}(\theta,\theta') = \bigg ( {N \over 2 \pi} \bigg )^2  -  
\bigg ({\sin \left(N (\theta - \theta')/2\right) \over 2 \pi  \sin \left((\theta - \theta')/2\right)} \bigg )^2 = 
\rho_{(1)}^{{\rm U}(N)}(\theta) \rho_{(1)}^{{\rm U}(N)}(\theta') +
\rho_{(2)}^{{\rm U}(N),T}(\theta,\theta').
\end{equation}
In the second equality
$\rho_{(1)}^{{\rm U}(N)}(\phi) = N/(2 \pi)$ independent of $\phi$,
while $\rho_{(2)}^{{\rm U}(N),T}(\theta,\theta')$ (referred to as the truncated two-point correlation) is the second term in the first equality of (\ref{A3}) and has the feature of being of order unity for $\theta \ne \theta'$ fixed. Generally the truncated two-point correlation is of interest as the kernel in the fluctuation formula
for the covariance of two linear statistics
\begin{equation}\label{A3a}
{\rm Cov} \Big ( \sum_{j=1}^N f(\theta_j), \sum_{j=1}^N g(\theta_j) \Big ) =
-{1 \over 2} \int_{-\pi}^{\pi} d \theta \int_{-\pi}^{\pi} d \theta' \,
( f(\theta) - g(\theta')) ( g(\theta) - f(\theta')) \rho_{(2)}^{T}(\theta,\theta');
\end{equation}
see e.g.~\cite[Prop.~2.1]{Fo23}.

Here we provide the exact evaluation of the two-point correlation for SU$(N)$. For this we will use
the functional differentiation method presented in \cite[Ch.~5.2.1]{Fo10}. To begin the calculation, we note that 
the Haar measure on U$(N)$, when written in terms of its eigenvalues
(see e.g.~\cite[Exercises 2.2 q.1]{Fo10}), 
 leads via the additional constraint $\sum_j \theta_j= 0~\mbox{mod}~2\pi$
 to the joint distribution function of
the eigenphases for SU$(N)$
\begin{align}
&\delta \Big (\sum_{j=1}^{N} \theta_j~\mbox{mod}~2\pi \Big )\cdot
\frac{1}{(2\pi)^{N-1}}
\frac{1}{N!}
|\Delta_N(\theta)|^2
=
\frac{1}{(2\pi)^N}
\frac{1}{N!}
\sum_{n=-\infty}^\infty e^{in(\theta_1+\cdots+\theta_{N})}
|\Delta_N(\theta)|^2,
\label{A4}\\
&\Delta_N(\theta):=
\prod_{1\leq j < k\leq N}
(e^{i \theta_k}-e^{i\theta_j})=\det[e^{ik \theta_j}]_{j=1,\ldots,N \atop k=0,\ldots,N-1}.
\nonumber
\end{align}
To obtain the first equality in (\ref{A4}), the Fourier sum form of the $2 \pi$-periodic Dirac delta function has been used.

Clearly, the $n=0$ term in (\ref{A4})
corresponds to the unconstrained U$(N)$ case.
We introduce a ``probe'' one-point function $a(\theta)$,
and define the generalized partition function by
\begin{equation}\label{A5}
Z_N[a]:=\mathbb{E}\Bigl[\prod_{j=1}^{N} a(\theta_j)\Bigr]=
\sum_{n=-\infty}^\infty
\frac{1}{N!}
\int_{-\pi}^\pi \!\!\cdots\!\!\int_{-\pi}^\pi
\prod_{j=1}^{N} \left(\frac{d\theta_j}{2\pi} a(\theta_j) e^{in\theta_j}\right)
|\Delta_N(\theta)|^2.
\end{equation}
The significance of this for present purposes is that
the $p$-point correlation functions are expressed as functional derivatives of $Z_N[a]$,
\begin{equation}
\rho_{(1)}^{{\rm SU}(N)}(\theta) = \left.
\frac{\delta Z_N[a]}{\delta a(\theta)} \right|_{a=1},
\quad
\rho_{(2)}^{{\rm SU}(N)}(\theta,\theta') = \left.
\frac{\delta^2 Z_N[a]}{\delta a(\theta)\delta a(\theta')}\right|_{a=1},
\label{A7}
\end{equation}
etc.

By expressing the Vandermonde determinant $\Delta_N(\theta)$
and its conjugate as sums of
permutations $P, Q\in \mathfrak{S}_N$, we can deduce that
each ($n^{\mathrm{th}}$) component of
$Z_N[a]$ takes the form of the Toeplitz determinant,
\begin{align}
Z_N[a]&=
\sum_{n=-\infty}^\infty
\frac{1}{N!}
\sum_{P,Q\in \mathfrak{S}_N}\sgn(P)\sgn(Q)
\prod_{j=1}^{N}
\int_{-\pi}^\pi
\frac{d\theta_j}{2\pi} a(\theta_j) e^{in\theta_j}
e^{iP(j)\theta_j}e^{-iQ(j)\theta_j}
\nonumber\\
&=\sum_{n=-\infty}^\infty
\sum_{R\in \mathfrak{S}_N}\sgn(R)
\prod_{j=1}^{N}
\int_{-\pi}^\pi
\frac{d\theta_j}{2\pi} a(\theta_j) e^{i(n+\ell-R(\ell))\theta_j}
~~~(\ell=P(j),~R=Q\circ P^{-1})
\nonumber\\
&=\sum_{n=-\infty}^\infty
\det\left[
\int_{-\pi}^\pi
\frac{d\theta}{2\pi} a(\theta) e^{i(n+m - \ell)\theta}
\right]_{\ell,m=1}^{N} 
\label{A8}
\end{align}
Here to obtain the final equality, the general fact that determinant of a matrix is equal to the determinant of the transpose of the matrix has been used.
Before using (\ref{A7}) to compute the two-point correlation,
it is instructive to first show how the first relation in  (\ref{A7}) can be used to
re-derive (\ref{A1}).

We have remarked that the $n=0$ term in 
(\ref{A4})
corresponds to the unconstrained U$(N)$ case. This implies that
the functional derivative operation in the first relation in (\ref{A7}) applied to
the $n=0$ term  in the final line of (\ref{A8}) yields
 $\rho_{(1)}^{{\rm U}(N)}(\theta)=\frac{N}{2\pi}$, as can be checked directly. Consider now the $n=1$ term
 of (\ref{A8}).
By applying the functional differentiation
$\frac{\delta}{\delta a(\theta)}$ at $a=1$
to each column of the determinant (it is generally true that a single functional derivative of a determinant can be carried out column-by-column), 
and using the simple fact that for $\ell, m , n \in \mathbb Z,$
\begin{equation}\label{A9}
\int_{-\pi}^\pi
\frac{d\theta}{2\pi} e^{i(n+\ell-m)\theta}
=\delta_{m-\ell,n},
\end{equation}
one sees that a non-zero contribution  comes from the final column ($m=N$),
\begin{equation}
\left|
\begin{array}{cccc}
0 & 0 & \cdots & e^{i(1+N-1)\theta}/2\pi\\
1 & 0 & \cdots & e^{i(1+N-2)\theta}/2\pi\\
\vdots  & \ddots & \ddots & \vdots \\
0  & \cdots  & 1 & e^{i(1+0)\theta}/2\pi
\end{array}
\right|
=(-1)^{N-1} \frac{e^{iN\theta}}{2\pi}.
\label{AX}
\end{equation}
By the same token, the $n=-1$ term of (\ref{A8}) contributes
the complex coujugate of (\ref{AX})
to $\rho_{(1)}^{{\rm SU}(N)}(\theta)$. On the other hand
terms with $|n|>1$ do not contribute as they
contain columns with entirely zero entries due to (\ref{A9}).
Thus for the one-point correlation function we obtain
\begin{equation}\label{A10}
\rho_{(1)}^{{\rm SU}(N)}(\theta) = \frac{N}{2\pi}
- (-1)^N \frac{e^{iN\theta}+e^{-iN\theta}}{2\pi},
\end{equation}
which is indeed consistent with (\ref{A1}).

According to (the second relation of \ref{A7}) two functional differentiation operations are required in (\ref{A8}). From the determinant structure, this is equivalent to taking the functional derivative of any two distinct columns, and summing over all choices.
One observes that the determinants with  $|n|>2$ do not contribute due to (\ref{A9}).
As the contribution from $n=0$ is guaranteed to give the U$(N)$ result (\ref{A3}),
one only needs to evaluate contributions from the terms with $|n|=1$ and $2$.
Easier is the $n=2$ contribution, in which 
$\frac{\delta}{\delta a(\theta)}$ and $\frac{\delta}{\delta a(\theta')}$ 
must act on the final two columns ($m=N-1, N$) to yield
\begin{equation}
\left|
\begin{array}{ccccc}
0 & 0 & \cdots &e^{i(2+N-2)\theta}/2\pi & e^{i(2+N-1)\theta'}/2\pi\\
0 & 0 & \cdots &e^{i(2+N-3)\theta}/2\pi & e^{i(2+N-2)\theta'}/2\pi\\
1 &  0 & \cdots &e^{i(2+N-4)\theta}/2\pi & e^{i(2+N-3)\theta'}/2\pi\\
\vdots  & \ddots & \ddots & \vdots & \vdots \\
0 &  \cdots & 1 &e^{i(2-1)\theta}/2\pi & e^{i(2+0)\theta'}/2\pi
\end{array}
\right|
=
\frac{
e^{iN(\theta+\theta')}
(1-e^{i(\theta'-\theta)})
}{(2\pi)^2}
\label{A11}
\end{equation}
and the one with $\theta$ and $\theta'$ exchanged.

For the $n=1$ term, 
one of $\frac{\delta}{\delta a(\theta)}$ and $\frac{\delta}{\delta a(\theta')}$
must act on the final column ($m=N$) but the other can act on any of
the remaining columns ($m=1,\ldots,N-1$).
The contribution from choosing the $m^{\mathrm{th}}$ column is
\begin{equation}
\left|
\begin{array}{ccccccc}
0 & 0 &  \cdots &e^{i(1+m-1)\theta}/2\pi & \cdots & 0 & e^{i(1+N-1)\theta'}/2\pi\\
1 & 0 &  \cdots &e^{i(1+m-2)\theta}/2\pi & \cdots & 0 & e^{i(1+N-2)\theta'}/2\pi\\
0 & 1 &  \cdots &e^{i(1+m-3)\theta}/2\pi & \cdots & 0 & e^{i(1+N-3)\theta'}/2\pi\\
\vdots &  \vdots & \ddots & \vdots & \cdots & \vdots & \vdots \\
0 &  0 & \cdots & 1/2\pi & \cdots & 0 & e^{i(1+N-m-1)\theta'}/2\pi \\
\vdots &  \vdots & \cdots & \vdots & \ddots & \vdots & \vdots \\
0 & 0 &  \cdots &e^{i(1+m-N)\theta}/2\pi & \cdots & 1 & e^{i(1+0)\theta'}/2\pi
\end{array}
\right|
=(-1)^{N}
\frac{e^{iN\theta'}
(e^{im(\theta-\theta')}-1)
}{(2\pi)^2}
\label{A12}
\end{equation}
and the one with $\theta$ and $\theta'$ exchanged.
By summing up contributions from $m=1,\ldots,N-1$, and
including complex-conjugated contributions from $n=-2$ and $-1$,
we finally obtain the two-point correlation function from (\ref{A11}) and (\ref{A12}),
\begin{multline}\label{A14}
\rho_{(2)}^{{\rm SU}(N)}(\theta,\theta') =
\rho_{(2)}^{{\rm U}(N)}(\theta,\theta')+
\frac{2}{\pi^2}
\cos N (\theta +\theta')
\sin^2\frac{\theta-\theta'}{2}
\\
+(-1)^N\frac{1}{\pi ^2}
\cos \frac{N(\theta+\theta')}{2}
\left(
\sin \frac{N(\theta -\theta')}{2}
\cot \frac{\theta -\theta'}{2}
-N \cos \frac{ N(\theta -\theta')}{2}\right).
\end{multline}
Here we employed the summation formula
\[
\sum_{m=1}^{N-1}
\sin m x \,\sin (N-m)x=
\frac{1}{2} (\sin N x \cot x -N \cos N x).
\]
After use of an elementary trigonometric identity, and recalling
(\ref{A1}) and (\ref{A3}), an alternative way to write (\ref{A14}) is seen to be
\begin{equation}
\rho_{(2)}^{{\rm SU}(N)}(\theta,\theta') = 
\rho_{(1)}^{{\rm SU}(N)}(\theta) \rho_{(1)}^{{\rm SU}(N)}(\theta') + 
\rho_{(2)}^{{\rm SU}(N),T}(\theta,\theta'),
\end{equation}
where
\begin{multline}\label{A16}
\rho_{(2)}^{{\rm SU}(N),T}(\theta,\theta') =
\rho_{(2)}^{{\rm U}(N),T}(\theta,\theta') +
\frac{2}{\pi^2}
\cos N (\theta +\theta')
\sin^2\frac{\theta-\theta'}{2}
\\
+(-1)^N\frac{1}{\pi ^2}
\cos \frac{N(\theta+\theta')}{2}
\sin \frac{N(\theta-\theta')}{2}
\cot \frac{\theta-\theta'}{2} 
- {1 \over \pi^2}
\cos N \theta \cos N \theta'. 
\end{multline}
Note that in (\ref{A16}), in distinction to (\ref{A14}), all the correction terms to the U$(N)$ result are of order unity.

In relation to its use in (\ref{A3a}), a strategy is to write the given $f,g$ as Fourier series, and similarly make use of the Fourier series expansion of $\rho_{(2)}^{{\rm SU}(N),T}$, as implied by the intermediate working in deriving (\ref{A16}). Another point of interest is the large $N$ expansion of (\ref{A16}) with bulk scaling, which is also a topic of present day interest for the $\beta$ generalisation of Dyson's circular ensembles
\cite{FS25}, and for spacing ratio distributions \cite{Ni25}.
Bulk scaling corresponds to the unfolding
$\theta = 2 \pi x/ N$, $\theta' = 2 \pi x'/N$ so that the mean eigenvalue spacing is unity. Recalling too (\ref{A3}), we compute
from (\ref{A16}) that 
\begin{multline}\label{A17}
(2 \pi /N)^2 \rho_{(2)}^{{\rm SU}(N),T}(2 \pi x/N, 2 \pi x'/N) \\
=
-\left(\frac{\sin \pi(x-x')}{\pi(x-x')}\right)^2
+\frac{4 (-1)^N}{N} \cos \pi(x+x') \frac{\sin \pi(x-x')}{\pi (x-x')}
+\mathcal{O}\left(\frac{1}{N^2}\right).
\end{multline}
It is notable that
the unit determinant constraint alters
the leading correction to the universal sine-kernel limit
(the first term of (\ref{A17})) 
from of $\mathcal{O}(N^{-2})$ to $\mathcal{O}(N^{-1})$,
which moreover decays as $1/(x-x')$ for a large separation of two eigenvalues.

\small

\providecommand{\bysame}{\leavevmode\hbox to3em{\hrulefill}\thinspace}
\providecommand{\MR}{\relax\ifhmode\unskip\space\fi MR }
\providecommand{\MRhref}[2]{%
  \href{http://www.ams.org/mathscinet-getitem?mr=#1}{#2}
}
\providecommand{\href}[2]{#2}

\end{document}